\UseRawInputEncoding
\documentclass[11pt,a4paper,english]{article}
\usepackage[utf8]{inputenc}
\usepackage[english]{babel} 
\usepackage[T1]{fontenc}        
\usepackage{textcomp}

\usepackage{amsmath,amsthm,amssymb,amsfonts}
\usepackage{enumitem,authblk} 
\usepackage{bm}
\usepackage{blkarray}

\usepackage{hyperref}

\usepackage{mathdots }

\usepackage{natbib}
\usepackage{booktabs, siunitx}

\usepackage{pdfpages}
\usepackage{graphicx}

\numberwithin{equation}{section}

\usepackage[super]{nth}

\usepackage[section]{placeins}

\newtheorem{prop}{Proposition}

\usepackage[section]{placeins}

\usepackage{bbm}

\newcommand{\deffeq}{\mathrel{\overset{\makebox[0pt]{\mbox{\normalfont\tiny\sffamily def}}}{=}}}

\newcommand{\R}{\mathbb{R}}
\newcommand{\E}{\mathbb{E}}

\newcommand{\U}{\mathcal{U}}

\newcommand{\eg}{{\textit{e.g.}, }}
\newcommand{\tth}{\footnotesize ^{\mbox{th}}\normalsize}
\newcommand{\ie}{{\it i.e., }}

\usepackage{tikz}
\usetikzlibrary{fit}
\usetikzlibrary{arrows.meta}
\tikzset{>={Latex[width=2mm,length=2mm]}}

\usepackage{outlines}

\renewenvironment{abstract}
{\small
	\begin{center}
		\bfseries \abstractname\vspace{-.5em}\vspace{0pt}
	\end{center}
	\list{}{
		\setlength{\leftmargin}{.0cm}%
		\setlength{\rightmargin}{\leftmargin}%
	}%
	\item\relax}
{\endlist}

\begin{document}

\pagestyle{myheadings}

\markboth{J.S. Pelck and R. Labouriau}
{J.S. Pelck and R. Labouriau}
\thispagestyle{empty}

\title{\Large
	Multivariate Generalised Linear Mixed Models
	With Graphical Latent Covariance Structure
}

\author[1] {Jeanett S. Pelck}

\author[1] {Rodrigo Labouriau 
	\thanks{Corresponding author: Rodrigo Labouriau, 
	rodrigo.labouriau@math.au.dk}}
\affil[1]{Department of Mathematics, Aarhus University, Denmark}
\date{July 2021}

\clearpage\maketitle
\thispagestyle{empty}

\begin{abstract}
This paper introduces a method for studying the correlation structure of a range of responses modelled by a multivariate generalised linear mixed model (MGLMM). 
The methodology requires the existence of clusters of observations and that each of the several responses studied is modelled using a generalised linear mixed models (GLMM) containing random components representing the clusters. 
We construct a MGLMM by assuming that the distribution of each of the random components representing the clusters is the marginal distribution of a (sufficiently regular) multivariate elliptically contoured distribution.
We use an undirected graphical model to represent the correlation structure of the random components representing the clusters of observations for each response. 
This representation allows us to draw conclusions regarding unknown underlying determining factors related to the clusters of observations.
Using a combination of an undirected graph and a directed acyclic graph (DAG), we jointly represent the correlation structure of the responses and the related random components. 
Applying the theory of graphical models allows us to describe and draw conclusions on the correlation and, in some cases, the dependence between responses of different statistical nature (\eg following different distributions, different linear predictors and link functions).
We present some simulation studies 
illustrating the proposed methodology.
\end{abstract}

\newpage

\tableofcontents

\newpage

\section{Introduction}
This paper introduces a method for studying the dependence structure of a range of responses modelled by a multivariate generalised linear mixed model (MGLMM) (see \citeauthor{Pelck2021A},  \citeyear{Pelck2021A}, for details). The methodology we suggest requires the existence of clusters of observations (or experimental units) and that each of the responses studied is modelled using a GLMM containing random components representing the clusters. We will construct an MGLMM by assuming that the distribution of each of the random components representing the clusters is the marginal distribution of a sufficiently regular multivariate elliptically contoured distribution (see  \citeauthor{Anderson2003}, \citeyear{Anderson2003}). This choice of the distribution of the random components includes, as a particular case, the multivariate normal distribution used in standard generalised linear mixed models (GLMMs), as the models considered in \cite{Breslow1993, Mcculloch2001, McCulloch1997}.

We will use an undirected graphical model (see \citealp{Lauritzen1996,Whittaker1990,Abreu2010}) to represent the correlation structure of the random components representing the clusters of observations for each response. This representation will allow us to conclude on unknown underlying determining factors related to the clusters of observations. Furthermore,  using a combination of an undirected graph and a directed acyclic graph (DAG), we jointly represent the correlation structure of the responses and the related random components. This representation arises naturally from the construction of the MGLMM we use and yields a known type of undirected graphical model, namely graphical models with block chain independence graph (BCG) as defined in \cite{Whittaker1990}. Remarkably, this construction will allow us to describe and draw conclusions on the correlation between the responses of different statistical nature, \eg responses modelled with GLMMs defined using various combinations of distributions, linear predictors and link functions (not necessarily the same for each response).

We will base the inference for the proposed graphical models on variants of tests for correlation under multivariate normality or multivariate elliptically contoured distributions studied in detail in \cite{Anderson2003}, from which we draw heavily. In the particular case where the random components are multivariate normally distributed, non-correlation will imply independence, which makes the conclusions of the analysis stronger. When the random components are not normally distributed but follow an elliptically contoured distribution, we will obtain slightly weaker conclusions since, in that case, lack of correlation implies only mean independence. 
\footnote{Recall that a random vector $X$ is mean independent of the random vector $Y$ when $E(X|Y=y) = E(X)$ for all $y$ in the support of the distribution of $Y$. It is well known that independence implies mean independence which implies non-correlation, but the reversed implications are in general not valid, see \cite{Wooldrige2010}.}

The paper is organised as follows. 
In Section~\ref{subsec:ModelDef}, we formulate a version of a multivariate generalised linear mixed model. For simplicity, we only specify one 
common clustering structure but the theory can easily be extended to include 
multiple clusterings.
In Section~\ref{subsec:GraphM}, we introduce essential 
concepts of graphical models that we will use to study the covariance structure of 
the random components and the response variables. These concepts are connected 
to the introduced multivariate model in Section~\ref{sec:conGraphGlmm}.
In Section~\ref{subsec:Test}, we describe statistical tests adapted from 
\cite{Anderson2003} and draw the connection to the theory of graphical models. 
In Section~\ref{subsec:SimStudy}, we perform a simulation study to study the 
distribution of the p-values in the simulated examples under the null hypothesis 
obtained using the statistical 
tests. Moreover, we study 
the power of the tests in the simulated examples on a grid consisting of values of 
the off-diagonal entry in the covariance matrix. Some concluding remarks are 
given in Section~\ref{sec:Con}.
Appendix~\ref{subsec:covEst} discusses how an estimate of 
the covariance matrix can be obtained based on consistent predictions of the 
random components. Appendix~\ref{App:testVden} presents details on how the 
density of the introduced test statistic can be evaluated in the case of Gaussian 
random components. 





\section{Multivariate Generalised Linear Mixed Models
	\label{sec:MulGLMM}}

In this section, we formulate a version of the multivariate generalised 
linear mixed model described in \cite{Pelck2021A}. 
These models are based on marginal GLMMs that extend the standard GLMMs in two directions: we assume the random components to be distributed according to an elliptical contoured distribution, instead of following a multivariate Gaussian distribution, and we assume the conditional distributions of each response, given the random components, to belong to a dispersion model, instead of an exponential dispersion model. The MGLMMs we will define require, however, the existence of clusters of observations and the presence of random components representing those clusters in each of the marginal GLMMs representing the responses. The result of this process is a rather flexible class of models that can be used in many practical applications; see for example \cite{pelck2020C,Pelck2021E,pelck2021D, Pelck2021F}.

\subsection{Model Definition
	\label{subsec:ModelDef}}

We define a $d$-dimensional multivariate generalised linear mixed 
model (\ie a MGLMM representing $d$ responses) with $n$ observations of 
the $j\tth$ marginal model, taking values in $\mathcal{Y}_j \subseteq \R$, for 
$j=1,\ldots,d$. Here $\mathcal{Y}_j$ is typically $\R$, $\R_+$, a compact real 
interval or $\mathbb{Z}_+$. 
Denote by  $Y_i^{[j]}$ the random variable representing the $i\tth$ observation of the $j\tth$ response, for $j=1,\ldots,d$ and $i=1, \ldots , n$.
 We assume that there exists a natural clustering of the observations causing dependence between observations arising from the same cluster (\eg grouping of observations within the same individual). 
We denote the cluster of the $i\tth$ observation, $Y_i^{[j]}$,  by $c(i)$ 
taking one of the values $1,\ldots,q$. Moreover, we assume that the clustering of the 
responses is 
independent of $d$, that is, the clusters are
represented in each marginal model. To ease the notation throughout the paper, we 
only consider one clustering mechanism but the methodology can be applied to a 
model with multiple clustering structures (\eg \cite{pelck2021D}).
In each marginal model, we consider random components each taking the same 
value for all responses within the corresponding cluster. These random 
components are denoted 
by $B^{[j]}_{c(i)}$ for $i=1,\ldots,n_j$ and $j=1,\ldots,d$.

Define the $d$-dimensional random vectors of random components, taking values 
in $\R^q$, by 
$\bm{B}^{[j]}=({B}_{1}^{[j]},\ldots,{B}_{q}^{[j]})^T$, the vector of 
responses $\bm{Y}^{[j]}=({Y}_{1}^{[j]},\ldots,{Y}_{n_j}^{[j]})$ and the 
vector of realisations $\bm{Y}^{[j]}$ by $\bm{y}^{[j]}$ (denoted 
\textit{observations}) of  
for 
$j=1,\ldots,d$. Moreover, consider the $q$-dimensional random vector 
$\bm{B}_l=({B}_{l}^{[1]},\ldots,{B}_{l}^{[d]})^T$ which we assume to be
elliptically contoured distributed \citep{Anderson2003} satisfying the 
following regularity conditions, for $l=1,\ldots,q$,
\begin{enumerate}
	\item The moments up to 
	fourth order of each marginal distribution exist
	\item Each of the marginal distributions is absolute continuous with respect to the Lebesgue measure
	\item All the conditional distributions exist and are elliptically 
	    contoured distributions also
	\item The location parameter vector is equal to zero.
\end{enumerate}
Furthermore, we assume that 
$\bm{B}_{l}$  is independent of $\bm{B}_{l'}$ for $l,l'=1,\ldots,q$ such that 
$l\neq l'$, \ie we assume that the random vectors representing different clusters 
are independent.
We define the density with respect to the Lebesgue measure of the elliptically contoured distribution by
\begin{align}
	\varphi (\bm{b} ;\bm{\Lambda})= \vert \bm{\Lambda}\vert^{-1/2} h\big(\bm{b}^T\bm{\Lambda}^{-1}\bm{b}\big), \label{eq:denEllip}
\end{align}
where $\bm{\Lambda}$ is a positive definite scatter matrix. The function $h(\cdot)$ is non-negative and satisfies that
\begin{align}
	\int_{\R^q} h(\bm{b}^T\bm{b})\,d\bm{b}=1. \label{eq:hInt}
\end{align}
When the density exists, the covariance matrix, $\bm{\Sigma}$, is proportional to 
$\bm{\Lambda}$, \ie the correlation matrix can be equivalently calculated from 
both $\bm{\Sigma}$ and  $\bm{\Lambda}$. An example of a commonly used 
distribution satisfying these regularity conditions is a multivariate Gaussian 
distribution with expectation zero and covariance matrix given by $\bm{\Sigma}$. 
Another example, that we will study later is the multivariate t-distribution. This 
distribution allows us to consider different degrees of tail heaviness. Note that 
because the 
moments of fourth order must exist in the multivariate t-distribution, the 
degrees of freedom should be larger than four.
 
According to the model, we assume that $Y_i^{[j]}$  is conditional distributed 
according to a dispersion model with dispersion parameter $\lambda_j\in 
\R_{+}$ given 
$B_{c(i)}^{[j]}$, and with conditional 
expectation
\begin{align*}
g_j\big(\mu_i^{[j]}(b)\big) \deffeq g_j(\E[Y_i^{[j]} \vert B_{c(i)}^{[j]}
=b]) 
= \bm{\beta}_j^T \bm{x}_{i}^{[j]}+{b}, \quad \forall \, b  \in \R,
\end{align*}
for all $i=1,\ldots,n_j$ and $j=1,\ldots ,d$. The vector $\bm{x}_i^{[j]}$ is a 
$p_j$ dimensional vector of explanatory variables corresponding to the vector of 
coefficients, $\bm{\beta}_j$. The explanatory variables might differ for the 
different responses. The function $g_j(\cdot)$ is a given link function, which is 
assumed to be strictly monotone, invertible and continuously differentiable. 
Below, we will suppress the dependence in $\mu_i^{[j]}(b)$ of $b$ to lighten the 
notation and denote the parameter space of the conditional means by 
$\mathcal{U}_j$.  
We define the conditional density corresponding to the conditional distribution of 
$Y_i^{[j]}$ given $B_{c(i)}^{[j]}=b$ with respect to a domination measure 
$\nu_j$ (defined on the measurable space $(\mathcal{Y}_j,\mathcal{A})$) by
\begin{align}\label{eq:fDenCondY}
	f(y_i^{[j]} \vert B_{c(i)}^{[j]}={b};\bm{\beta}^{[j]},\lambda_j)\deffeq p(y_i^{[j]};\mu_i^{[j]},\lambda_j)=
	a_j(y_i^{[j]};\lambda_j)\exp [-\tfrac{1}{2\lambda_j} \, 
	d_j(y_i^{[j]};\mu_i^{[j]})].
\end{align} 
The function $d_j: \mathcal{Y}_j\times \U_j \rightarrow \R_+$ is the \emph{unit 
	deviance} and, by definition, satisfies that $d_j(\mu ,\mu)=0$ and 
	$d_j(y,\mu)>0$ 
for all $(y,\mu)\in \mathcal{Y}_j\times \U_j$ such that $ y\neq \mu$. The 
function 
$a_j:\mathcal{Y}_j\times \R_+ \rightarrow \R_+$ is a given normalising 
function. 
We assume that the unit deviance is regular, that is, $d$ is twice continuously 
differentiable in $\mathcal{Y}_j\times  \U_j$ and 
$\partial^2d(\mu;\mu)/\partial\mu^2>0$ for all $\mu\in\U_j$. The function 
$V_j:\U_j\rightarrow\R_+$ given by $V_j(\mu)=2/\{\partial^2 
d_j(\mu,\mu)/\partial\mu^2\}$  for all $\mu$ in $\U_j$ is termed the 
\emph{variance 
	function} \citep{Cordeiro2021}. 
The following families of distributions are examples of dispersion models: 
Normal, Gamma, inverse Gaussian, von Mises, Poisson, and Binomial families.
This setup defines a version of the multivariate GLMM described in  
\cite{Pelck2021A} with the additional assumption that the multivariate 
distribution of the random components follow an elliptical contoured distribution.

%

%
%


\section{Representation of the Latent Covariance Structure via Graphical Models
	\label{sec:LatentStr}}
In this section, we describe and illustrate how we can use the theory of graphical 
models to examine the latent covariance structure 
of the random components in the multivariate model described above, and how 
this covariance structure affects the correlation between the responses. First, we 
give a short account for the theory of graphical models. For a more 
comprehensively description see \cite{Lauritzen1996} and \cite{Whittaker1990}.

\subsection{Basic Theory of Graphical Models
    \label{subsec:GraphM}}

Let $\mathcal{G}=(\mathcal{V},\mathcal{E})$ denote a graph defined with a set 
of vertices, $\mathcal{V}$, composed of random variables and a set of edges, 
$\mathcal{E}\in \mathcal{V} \times \mathcal{V}$. The set of edges, 
$\mathcal{E}$, consists of pairs of elements taken from $\mathcal{V}$. We 
distinguish between undirected independence graphs (UGs) and directed acyclic 
independence graphs (DAGs) but the two types of graphs can be combined as we 
will see below. The two types of graphs differ because of the underlying 
assumption of symmetry in the roles played by the variables in an UG, whereas in 
a DAG one variable can carry information on another without the converse being necessarily true. In the DAG we use an arrow from one variable pointing to another variable to indicate that the first variable carries information on the second.
In an UG, two vertices are connected by an edge if, and 
only if, they are not conditionally independent given the remaining variables in 
$\mathcal{V}$. This is the same definition used for DAGs with the conditioning 
set modified from the remaining variables to a set containing all remaining 
variables that carry information on one of the two vertices either direct or 
through the other vertices in $\mathcal{V}$.

In an UG, we say that there is a path connecting two vertices, say $v_1$ and 
$v_n$, if there 
exists a sequence of vertices $v_1, \ldots, v_n$ such that, for $i = 1, \ldots , n-1$, 
the pair $(v_i, v_{i+1})$ is in $\mathcal{E}$. A set of vertices $S$, separates 
two disjoint sets of 
vertices $A$ and $B$ in the graph  $\mathcal{G}=(\mathcal{V},\mathcal{E})$ 
when every path connecting a vertex in $A$ to a vertex in $B$ necessarily 
contains a vertex in $S$.  According to the theory of graphical models (see 
\citeauthor{Lauritzen1996},  \citeyear{Lauritzen1996} and 
\citeauthor{Perl2009}, \citeyear{Perl2009}), the UG defined above satisfies the 
\emph{separation principle}, which states that if a set of vertices $S $, separates 
two disjoint subsets of 
vertices $A $ and $B$ in the graph $\mathcal{G}=(\mathcal{V},\mathcal{E})$, 
then all variables in 
$A$ are independent of all variables in $B$ given $S$.
Moreover, if the subsets $A$ and $B$ 
are isolated (\ie there are no paths connecting a vertex in $A$ to a vertex in $B$), 
then the variables in $A$ are independent of  the variables in $B$.

A DAG possesses the Markov properties of its associated moral graph. Here the associated moral graph of a DAG is the UG obtained by the same vertex set but with a modified set of edges. The modified set of edges is formed by all the existing edges in the DAG replaced by undirected edges together with all edges necessary to eliminate forbidden Wermuth configurations. The latter means that for each vertex, we connect all vertices that have a directed edge towards the vertex in question with an undirected edge.

The two types of graphs can be combined into a block chain independence graph (BCG). In this graph, we assume that the vertex set $\mathcal{V}$ can be partitioned into subsets, called blocks, which are connected by directed edges but where all edges within the same block are undirected. As for the DAG, the BCG processes the same independence interpretation as its associated moral graph. For more information see \cite{Lauritzen1996} and \cite{Whittaker1990}.

\subsection{Connecting the Multivariate Model with the Theory of Graphical Models
\label{sec:conGraphGlmm}}
We connect the model formulated in Section~\ref{subsec:ModelDef} with the 
theory of graphical models by defining an undirected graph  
$\mathcal{G}=(\mathcal{V},\mathcal{E})$, with 
$\mathcal{V}=\{\bm{B}^{[1]},\ldots,\bm{B}^{[d]}\}$, where 
$\bm{B}^{[1]},\ldots,\bm{B}^{[d]}$  are the vectors of random components in 
the multivariate model described in Section~\ref{subsec:ModelDef}. In this 
context, the edges can only be interpret in terms of independence when the 
random components are Gaussian distributed. In the case of a non-Gaussian elliptically contoured distribution, two vertices are connected by an edge if, and only if, they 
are conditionally correlated given the remaining variables, which in this context implies conditional mean independence.
The set of vertices can also be formulated in terms of each variable in the model instead of vectors as above. In this case, the graphical representation will consist of $q$ separated cliques each containing the respective entry of the vectors $\bm{B}^{[1]},\ldots,\bm{B}^{[d]}$ due to the model assumptions. The choice of representation depends on the analysis and which choice that leads to the best discussion of the results. Note, that the results does not change only the visualisation. We will consider the vector representation below. 

The graph defined above is interpret in terms of the random components as 
follows: if, for example, $\bm{B}^{[1]}$ and $\bm{B}^{[2]}$ are connected 
with an edge, then these two random variables are conditionally correlated 
given $\{\bm{B}^{[3]},\ldots,\bm{B}^{[d]}\}$. Therefore, $\bm{B}^{[1]}$ 
carries some information on $\bm{B}^{[2]}$ not contained in the other variables. 
For example if the random components represent variation between different 
blocks in a field experiment, this means that there are some latent factors 
affecting the blocks, could be some characteristics of the soil, which affect the 
first and second response differently than the other responses.

We introduce an extension of the separation principle below, which we call the \textit{induced separation principle}. This can be used to draw general conclusions on the response variables. 
According to the model, the responses are independent given the random 
components. Therefore, conditional independence/un-correlation between, say, 
$\bm{B}^{[1]}$ and $\bm{B}^{[2]}$  given  
$\{\bm{B}^{[3]}, \allowbreak \ldots ,\allowbreak \bm{B}^{[d]}\}$ imply that $\bm{Y}^{[1]}$ and 
$\bm{Y}^{[2]}$ are conditionally uncorrelated given 
$\{\bm{B}^{[3]},\ldots,\bm{B}^{[d]}\}$.
By including the random variables $Y_i^{[j]}$ (for $j=1, \ldots, d$ and $i = 1,\ldots n$) in the set of vertices, and by taking the 
model assumptions into considerations, it is possible to formulate a block chain 
independence graph that represents the covariance structure both among the 
random components but also within the response variables. The theory of BCG 
makes it possible to extend the separation principle to a version that applies to the 
total graph including both the random components and the response variables. That 
is, by looking at the moral graph, we can determine all conditional 
uncorrelations \citep[Theorem. 3.6.1]{Whittaker1990}.

We will describe how the BCG can be constructed in the multivariate model 
described in Section~\ref{subsec:ModelDef}. For simplicity we only consider 
one common clustering mechanism in this model, however, below we will argue 
how the BCG can be constructed in the case of multiple random components.
We define a block chain independence graph $\mathcal{G}'=(\mathcal{V}',\mathcal{E}')$ \citep{Whittaker1990} by letting $\mathcal{V}'=\mathcal{V}\cup \{\bm{Y}^{[1]},\ldots,\bm{Y}^{[d]}\}$  and $\mathcal{E}'=\mathcal{E}\cup \mathcal{E}_{{Y}}$, where $\mathcal{V}$  and $\mathcal{E}$ are defined as above. Here, $\mathcal{E}_{{Y}}$ include directed edges from $\bm{B}^{[j]}$ to $\bm{Y}^{[j]}$ for $j=1,\ldots,d$, whereas $\mathcal{E}$ only include 
undirected edges. Usually, the way to separate undirected and directed edges in $\mathcal{E}'$ is to use the notation that if there is an directed edge from $V_i$ to $V_j$, the edge $(i,j)$ is included in $\mathcal{E}'$. However, if there is an undirected edge from $V_i$ to $V_j$ both the edge $(i,j)$ and $(j,i)$ are included in $\mathcal{E}'$.
The essential property of this graph is that by construction, any edge is undirected for intra-block vertices, and directed for inter-block vertices with direction from the random components to the response variables (the blocks are here defined by $\mathcal{V}$ and  $\mathcal{V}_{\bm{Y}}=\{\bm{Y}^{[1]},\ldots,\bm{Y}^{[d]}\}$). 
The induced separation principle implies that if $\mathcal{S}$ separates two 
disjoint subsets of vertices, $\mathcal{A}$ and $\mathcal{B}$ in $\mathcal{V}$, 
and $\mathcal{A}'$ and $\mathcal{B}'$ are the sets of the corresponding 
response variables, respectively, then all response variables in $\mathcal{A}'$  are 
conditionally uncorrelated of the variables in $\mathcal{B}'$ given 
the random components in $\mathcal{S}$.

In the case of multiple clustering mechanisms, we redefine $\mathcal{V}'$ to be 
the union of all sets of random components and the responses, that is,  
$\mathcal{V}'=\mathcal{V}_1\cup \ldots \cup \mathcal{V}_b \cup  
\mathcal{V}_{\bm{Y}}$, where $b$ is the total number of clustering 
mechanisms, 
and $\mathcal{V}_i$ is the set containing random vectors corresponding to the 
random components associated with the $i\tth$ clustering. The edges in this graph 
consist of 
the undirected edges inside each block together with directed edges from each 
random vector pointing 
towards the corresponding response variable (between the blocks). Under the 
model, we assume that 
each block of random components is independent of the others. Therefore, we do 
not need to connect 
the blocks $\mathcal{V}_1,\ldots,\mathcal{V}_b$ with an edge. This
structure is illustrated in Figure~\ref{fig:chainGraphGenerel}. 

The moral version of such a graph can be difficult to interpret in terms of 
conditional uncorrelation between the response variables. In that 
case, we suggest 
to either only consider the undirected graphs for the random components 
excluding the response vectors, or if one of the clustering mechanisms are of 
particular interest, we can restrict ourself to only examining the graph that 
includes the random components and the response variables of interest. In the 
latter case, we are only able to interpret the graph on individual level. For 
example, in a study with two clustering mechanisms: one representing individual 
variation and another clustering the individuals in different groups, we might only 
be interested in examining the correlation between different 
responses caused by the individual clustering structure. Therefore, we can 
consider a graphical representation of the covariance structure of the individual 
variation for each individual and thus, avoid comparing individuals within the same 
group for which the corresponding responses will be correlated do to the random 
component grouping the individuals. Thus, in the complete block chain 
independence graph, many of the responses will only be conditional 
uncorrelated after conditioning on multiple clusterings. 

An example of a block chain independence graph representing a three dimensional 
model with two clusterings and  it's corresponding moral graph is presented in 
Figure~\ref{fig:chainGraphExample}. In this example we observe from the moral 
graph that $\bm{Y}^{[1]}$ and $\bm{Y}^{[3]}$ are conditionally 
uncorrelated given $\bm{B}_1^{[2]}$ and  $\bm{B}_2^{[2]}$. If 
there was an edge connecting $\bm{B}_1^{[2]}$ and $\bm{B}_3^{[2]}$, then 
$\bm{Y}^{[1]}$ and $\bm{Y}^{[3]}$ would only be conditionally independent 
given all the random components.

 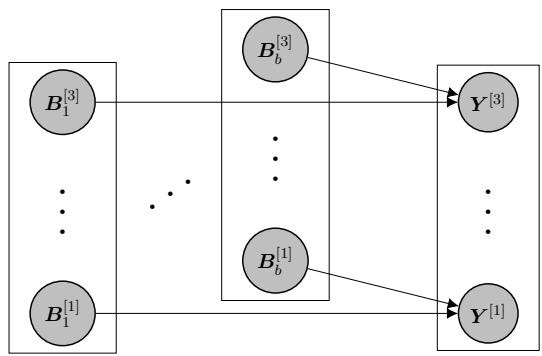
\begin{figure}[htbp!]
 	\centering
 	\scalebox{0.7}{
 		\begin{tikzpicture}[
 			squarednode/.style={rectangle, draw=black, fill=lightgray, thick},
 			roundnode/.style={circle, draw=black, fill=lightgray, thick}]
 			\node [style=roundnode,minimum size=1cm] (B1) at (-2, 0) {$\bm{B}^{[1]}_1$};
 			\node [style=roundnode,minimum size=1cm] (B3) at (-2, 4) {$\bm{B}^{[3]}_1$};
 			\path (B1) -- (B3) node [font=\Huge, midway, sloped] {$\dots$};
 			\node[draw,inner xsep=1em,fit=(B1)(B3)] {};
 			\node [style=roundnode,minimum size=1cm] (U1) at (2, 1) {$\bm{B}^{[1]}_b$};
 			\node [style=roundnode,minimum size=1cm] (U3) at (2, 5) {$\bm{B}^{[3]}_b$};
 			\path (U1) -- (U3) node [font=\Huge, midway, sloped] {$\dots$};
 			\node[draw,inner xsep=1em,fit=(U1)(U3)] {}; 
 			\node [style=roundnode,minimum size=1cm] (Y1) at (6, 0) {$\bm{Y}^{[1]}$};
 			\node [style=roundnode,minimum size=1cm] (Y3) at (6, 4) {$\bm{Y}^{[3]}$};
 			\path (Y1) -- (Y3) node [font=\Huge, midway, sloped] {$\dots$};
 			\node[draw,inner xsep=1em,fit=(Y1)(Y3)] {};
 			\draw[->] (B1) to (Y1);
 			\draw[->] (B3) to (Y3);
 			\draw[->] (U1) to (Y1);
 			\draw[->] (U3) to (Y3);
 			\node [font=\Huge]  at (0,2.5) {$\iddots$};
 	\end{tikzpicture}}
 	\caption{Illustration of the structure of a block chain independence graph, ignoring the undirected edges inside each of the $(b+1)$ blocks. }
 	\label{fig:chainGraphGenerel}
 \end{figure}

 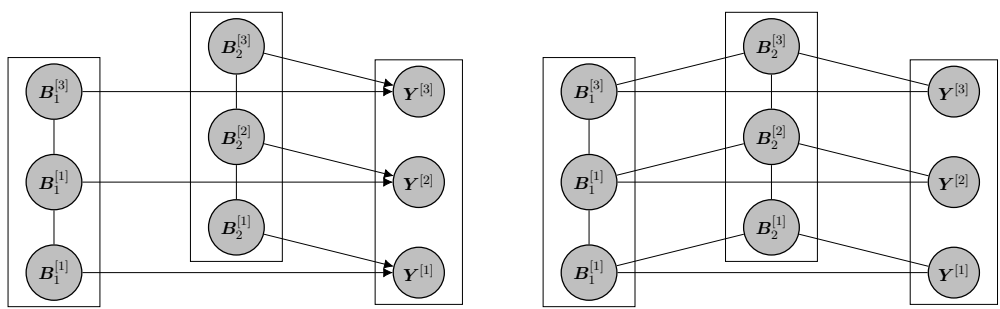
\begin{figure}[htbp!]
 	\centering
 	\scalebox{0.6}{
 		\begin{tikzpicture}[
 			squarednode/.style={rectangle, draw=black, fill=lightgray, thick},
 			roundnode/.style={circle, draw=black, fill=lightgray, thick}]
 			\node [style=roundnode,minimum size=1cm] (B1) at (-2, 0) {$\bm{B}^{[1]}_1$};
 			\node [style=roundnode,minimum size=1cm] (B2) at (-2, 2) {$\bm{B}^{[1]}_1$};
 			\node [style=roundnode,minimum size=1cm] (B3) at (-2, 4) {$\bm{B}^{[3]}_1$};
 			\node[draw,inner xsep=1em,fit=(B1)(B3)] {};
 			\node [style=roundnode,minimum size=1cm] (U1) at (2, 1) {$\bm{B}^{[1]}_2$};
 			\node [style=roundnode,minimum size=1cm] (U2) at (2, 3) {$\bm{B}^{[2]}_2$};
 			\node [style=roundnode,minimum size=1cm] (U3) at (2, 5) {$\bm{B}^{[3]}_2$};
 			\node[draw,inner xsep=1em,fit=(U1)(U3)] {}; 
 			\node [style=roundnode,minimum size=1cm] (Y1) at (6, 0) {$\bm{Y}^{[1]}$};
 			\node [style=roundnode,minimum size=1cm] (Y2) at (6, 2) {$\bm{Y}^{[2]}$};
 			\node [style=roundnode,minimum size=1cm] (Y3) at (6, 4) {$\bm{Y}^{[3]}$};
 			\node[draw,inner xsep=1em,fit=(Y1)(Y3)] {};
 			\draw[->] (B1) to (Y1);
 			\draw[->] (B2) to (Y2);
 			\draw[->] (B3) to (Y3);
 			\draw[->] (U1) to (Y1);
 			\draw[->] (U2) to (Y2);
 			\draw[->] (U3) to (Y3);
 			\draw[-] (B1) to (B2);
 			\draw[-] (B2) to (B3);
 			\draw[-] (U1) to (U2);
 			\draw[-] (U2) to (U3);
 	\end{tikzpicture}
  	\hspace{1.5cm}
  	\begin{tikzpicture}[
 	squarednode/.style={rectangle, draw=black, fill=lightgray, thick},
 	roundnode/.style={circle, draw=black, fill=lightgray, thick}]
 	\node [style=roundnode,minimum size=1cm] (B1) at (-2, 0) {$\bm{B}^{[1]}_1$};
 	\node [style=roundnode,minimum size=1cm] (B2) at (-2, 2) {$\bm{B}^{[1]}_1$};
 	\node [style=roundnode,minimum size=1cm] (B3) at (-2, 4) {$\bm{B}^{[3]}_1$};
 	\node[draw,inner xsep=1em,fit=(B1)(B3)] {};
 	\node [style=roundnode,minimum size=1cm] (U1) at (2, 1) {$\bm{B}^{[1]}_2$};
 	\node [style=roundnode,minimum size=1cm] (U2) at (2, 3) {$\bm{B}^{[2]}_2$};
 	\node [style=roundnode,minimum size=1cm] (U3) at (2, 5) {$\bm{B}^{[3]}_2$};
 	\node[draw,inner xsep=1em,fit=(U1)(U3)] {}; 
 	\node [style=roundnode,minimum size=1cm] (Y1) at (6, 0) {$\bm{Y}^{[1]}$};
 	\node [style=roundnode,minimum size=1cm] (Y2) at (6, 2) {$\bm{Y}^{[2]}$};
 	\node [style=roundnode,minimum size=1cm] (Y3) at (6, 4) {$\bm{Y}^{[3]}$};
 	\node[draw,inner xsep=1em,fit=(Y1)(Y3)] {};
 	\draw[-] (B1) to (Y1);
 	\draw[-] (B2) to (Y2);
 	\draw[-] (B3) to (Y3);
 	\draw[-] (U1) to (Y1);
 	\draw[-] (U2) to (Y2);
 	\draw[-] (U3) to (Y3);
 	\draw[-] (B1) to (B2);
 	\draw[-] (B2) to (B3);
 	\draw[-] (U1) to (U2);
 	\draw[-] (U2) to (U3);
    \draw[-] (B1) to (U1);
    \draw[-] (B2) to (U2);
    \draw[-] (B3) to (U3);
\end{tikzpicture}}
 	\caption{Example of a BCG for a three dimensional multivariate generalised 
 	linear mixed model with two random components, and it's corresponding moral 
 	graph (to the right).}
 	\label{fig:chainGraphExample}
 \end{figure}

 \FloatBarrier
\subsection{Testing the Covariance Structure
	\label{subsec:Test}}
In this section, we formulate a statistical test based on the results in 
\cite{Anderson2003}. Using this test, it is possible to test for (conditional) 
uncorrelation between pairs or groups of variables. 

We introduce some general notation that we will use to describe the statistical 
test in the case where the random components are assumed to be Gaussian 
distributed and the more general setup where we assume an elliptical contoured 
distribution. In both cases, we can test for uncorrelation between 
groups of variables either directly or conditional on a separating set. We show the 
conditional test but the approach is equivalent in the direct case.

Let $\bm{X}=(\bm{X}^T_{(1)},\ldots, \bm{X}^T_{(k)})^T$ be a $p$-dimensional random vector distributed according to an elliptically contoured distribution (including the special case of a Gaussian distribution) with location parameter equal to zero and a positive definite scatter matrix 
\begin{align*}
	\bm{\Lambda}= \begin{bmatrix}
		\bm{\Lambda}_{11} & \bm{\Lambda}_{12} & \cdots & \bm{\Lambda}_{1k}\\
		\bm{\Lambda}_{21} & \bm{\Lambda}_{22} & \cdots & \bm{\Lambda}_{2k}\\
		\vdots           &   \vdots         &        &  \vdots \\
		\bm{\Lambda}_{k1} & \bm{\Lambda}_{k2} & \cdots & \bm{\Lambda}_{kk}
	\end{bmatrix}=
\begin{bmatrix}
		\tilde{\bm{\Lambda}}^{(k-1)} & \tilde{\bm{\Lambda}}^{(k-1,k)}\\
\tilde{\bm{\Lambda}}^{(k,k-1)} & {\bm{\Lambda}}_{kk}
\end{bmatrix},
\end{align*}
which is proportionel to the covariance matrix $\bm{\Sigma}$.
We assume that the density of $\bm{X}$ exists with respect to the Lebesgue measure. Moreover, we assume that the conditional distribution of  $\bm{X}^{(k-1)}=(\bm{X}^T_{(1)},\ldots, \bm{X}^T_{(k-1)})$ given $\bm{X}_{(k)}$ exists. The distribution of   $\bm{X}^{(k-1)} \vert \bm{X}_{(k)}$ is also elliptically contoured distributed with scatter matrix
\begin{align}
	\bm{\Lambda}_{\cdot k}=\tilde{\bm{\Lambda}}^{(k-1)}-\tilde{\bm{\Lambda}}^{(k-1,k)} \bm{\Lambda}_{kk}^{-1}\tilde{\bm{\Lambda}}^{(k,k-1)} , \label{eq:condForm}
\end{align}
which is proportional to the covariance matrix in the conditional distribution 
\citep{Anderson2003}. Consequently, the formulas which apply in the normal 
case apply in this more general  setting as well. 

We would like to test the null hypothesis that the subvectors $\bm{X}_{(1)},\allowbreak \ldots, \allowbreak \bm{X}_{(k-1)}$ are independent given $\bm{X}_{(k)}$. This is equivalent to examining if $\bm{\Lambda}_{\cdot k}$ is on the form
\begin{align*}
	\bm{\Lambda}_{\cdot k}= \begin{bmatrix}
		\bm{\Lambda}_{11\cdot k} & \bm{0} & \cdots & \bm{0}\\
		\bm{0} & \bm{\Lambda}_{22\cdot k} & \cdots & \bm{0}\\
		\vdots           &   \vdots         &        &  \vdots \\
		\bm{0} & \bm{0} & \cdots & \bm{\Lambda}_{(k-1)(k-1)\cdot k}
	\end{bmatrix}.
\end{align*}

We first treat the special case where $\bm{X}$ is Gaussian distributed below. In this case, the statistical test will be exact. Second, we present an asymptotic test when the number of realisations of $\bm{X}$, denoted $q$, goes to infinity which  is valid in the case of a general elliptically contoured distribution.

\subsubsection{Normally Distributed Random Component
\label{subsec:normalTest}}

Here, we show a test for conditional independence for subsets of variables in a 
Gaussian distributed vector.
Let $\bm{A}_{\cdot k}$ be the maximum likelihood estimate of  
$\bm{\Lambda}_{\cdot k}$ or another estimate proportional to the maximum 
likelihood estimate based on $q$ observations ($\bm{A}_{\cdot k}$ can also be 
calculated from a maximum likelihood estimate of $\bm{\Sigma}$ using the 
formula in \eqref{eq:condForm}).

The test statistic we will consider is given by 
\begin{align*}
	V=\frac{\det (\bm{A}_{\cdot k})}{\prod_{i=1}^{k-1} \det(\bm{A}_{ii \cdot k})},
\end{align*}
that is, $V=\lambda^{q/2}$ where $\lambda$ is the likelihood ratio statistic and $q$ the number of observations (in the setup of multivariate GLMMs this is the number of groups of the random component). 

It can be shown that under the null hypothesis  \citep{Anderson2003} the 
distribution of $V$ is given by
\begin{align}
	V \sim \prod_{i=2}^{k-1}\prod_{j=1}^{d_i} Z_{ij}, \label{eq:testV}
\end{align}
where the random variables $Z_{21},\ldots,Z_{(k-1)d_{k-1}}$ are independent and $Z_{ij}\sim \text{Beta}  \big (\tfrac{1}{2}[q-\bar{d}_i-j], \tfrac{1}{2} \bar{d}_{i} \big)$ with $\bar{d}_i=d_1+\ldots+d_{i-1}$ for $i=2,\ldots,(k-1)$ and $j=1,\ldots,d_i$.


The continuity of the determinant function implies that $V$ remains constant for any estimated scatter matrix proportional to the maximum likelihood estimate, and thus the distribution is still exact. For a consistent estimator of the covariance matrix or scatter matrix, the distribution is only asymptotic.

\subsubsection{Elliptical Contoured Distributed Random Component
\label{subsubsec:generalTest}}
Under the assumption of a general elliptical contoured distribution, 
\cite{Anderson2003} shows that the following test statistic can be used to test 
asymptotically if the correlation between groups of variables are zero for $q$ 
going to 
infinity (either direct or conditioning on a separating set).

Let $\bm{A}_{\cdot k}$ denote the sample estimate of the covariance matrix of $\bm{X}^{(k-1)}$ given $\bm{X}_k$. This can be estimated directly or calculated using the formula in \eqref{eq:condForm} on the sample covariance matrix given by
\begin{align*}
	\bm{A}= \frac{1}{q-1}\sum_{i=1}^q (\bm{x}_i-\bar{\bm{x}})(\bm{x}_i-\bar{\bm{x}})^T.
\end{align*}
Define $\tilde{\bm{A}}_{\cdot k}^{(i-1)}$ as the $\bar{d}_i\times \bar{d}_i$ dimensional sub-matrix of $\bm{A}_{\cdot k }$ with the first $\bar{d}_i=d_1+\dots+d_{i-1}$ rows and columns of $\bm{A}_{\cdot k}$. Moreover, let $\tilde{\bm{A}}_{\cdot k}^{(i,i-1)}$ denote the sub-matrix corresponding to the first $\bar{d}_i$ columns and the rows $\bar{d}_i+1,\dots, \bar{d}_i+d_i$ of $\bm{A}_{\cdot k }$.
Define
\begin{align*}
	\bm{H}_i& = (q-1)\tilde{\bm{A}}_{\cdot k}^{(i,i-1)} \tilde{\bm{A}}_{\cdot k}^{(i-1)} (\tilde{\bm{A}}_{\cdot k}^{(i,i-1)})^T\\
	\bm{G}_i&= (q-1) \bm{A}_{ii\cdot k}-\bm{H}_i\\
	V_i &= \frac{\vert \bm{G}_i \vert}{\vert \bm{G}_i+\bm{H}_i \vert},
\end{align*}
where $\bm{A}_{ii\cdot k}$ is the $(i,i)^{th}$ block matrix of $\bm{A}_{\cdot k}$.

The test statistic for the null hypothesis that $(\bm{X}_1,\ldots,\bm{X}_{k-1})$ are conditionally independent given $\bm{X}_k$ is formulated as
\begin{align*}
	-q\sum_{i=2}^{k-1} \log V_i,
\end{align*}
which converges in distribution to $(1+\kappa)\chi^2_f$ for $q$ going to infinity, \\
$f=\sum_{i=2}^{k-1} \bar{d}_i d_i$ and
\begin{align*}
	\kappa = (\bar{d}_k(\bar{d}_k+2))^{-1}\E[(\bm{X}^{(k-1)})^T\bm{\Sigma}_{\cdot k}^{-1}\bm{X}^{(k-1)}]^2-1.
\end{align*}
We can estimate the kurtosis parameter by
\begin{align*}
	\hat{\kappa}=(\bar{d}_k(\bar{d}_k+2))^{-1}\frac{1}{q}\sum_{i=1}^q [(\bm{x}_i^{(k-1)})^T\bm{A}_{\cdot k}^{-1}\bm{x}_i^{(k-1)}]^2-1.
\end{align*}

\subsubsection{Simulation Study of Convergence Rate 
\label{sec:simStudConvRat}}
In this section we study the power of the introduced tests as a function of $q$ in a 
simulated example, that is, the probability of accepting the hypothesis of 
independence/un-correlation when it is true. Working with a five percent 
significance level this 
should be close to $95$ percent when $q$ is large enough.  
For different values of  $q$, we simulate $10\,000$ times $q$ random variables 
from a $4$ dimensional Gaussian and t-distribution, respectively, with 
expectation zero and covariance matrix
\begin{align*}
	\bm{\Sigma}=
	\begin{pmatrix}
		 0.4083  & 0.000000 & 0.000000  & 0.000000\\
		0.0000   & 0.456510  & -0.451965 & 0.265170\\
		0.0000   & -0.451965  & 0.837030  & -0.491090\\
		0.0000   & 0.265170   & -0.491090  & 0.524365
	\end{pmatrix}.
\end{align*}
For each of the $10\,000$ simulations, we estimate the sample covariance matrix 
and apply the appropriate test described above. Based on the calculated p-values, 
we can examine how many times the estimated p-value is above five percent and 
divide this by the number of simulations in  order to obtain an estimate of the 
probability of accepting the hypothesis of independence/un-correlation when it is 
true. If the 
distributional assumption of the test statistic is correct, the estimated value will 
be close to $95$ percent when working with a five percent significance level. 

We formulate the simulated model formally by letting 
$\bm{X}=(X_1,\ldots,X_4)$ and $\bm{Y}=(Y_1,\ldots,Y_4)$ denote random 
vectors distributed according to a multivariate Gaussian and t-distribution, 
respectively, with mean zero and covariance matrix $\bm{\Sigma}$. In the 
multivariate t-distribution, the degrees of freedom is assumed to be five.  Let 
$\bm{X}_1,\ldots,\bm{X}_{q}$ and $\bm{Y}_1,\ldots,\bm{Y}_{q}$ denote 
i.i.d. copies of $\bm{X}$ and $\bm{Y}$, respectively, corresponding to the 
simulated random variables in each of the $10\,000$ rounds. Using the 
realizations of these variables denoted $\bm{x}_1,\ldots,\bm{x}_{q}$ and 
$\bm{y}_1,\ldots,\bm{y}_{q}$, we estimate in each of the $10\,000$ rounds 
$\bm{\Sigma}$ by
\begin{align*}
	\hat{\bm{\Sigma}}_{\bm{x}}& = \tfrac{1}{q}\sum_{i=1}^{\scriptscriptstyle{q}} (\bm{x}_i-\bar{\bm{x}})(\bm{x}_i-\bar{\bm{x}})^T\\
		\hat{\bm{\Sigma}}_{\bm{y}}& = \tfrac{1}{q}\sum_{i=1}^{\scriptscriptstyle{q}} (\bm{y}_i-\bar{\bm{y}})(\bm{y}_i-\bar{\bm{y}})^T,
\end{align*}
where $\bar{\bm{x}}$ and $\bar{\bm{y}}$ are the estimated means. In each of 
the the $10\,000$ rounds, we test the hypothesis that
\begin{align*}
	H: \, \bm{\Sigma}_{13} = \bm{\Sigma}_{31} = 0,
\end{align*}
for both the multivariate Gaussian and t-distribution based on the estimated covariance matrices $\hat{\bm{\Sigma}}_{\bm{x}}$ and $\hat{\bm{\Sigma}}_{\bm{y}}$, respectively, with $\bm{\Sigma}_{ij}$ denoting the $(i,j)\tth$ entry in $\bm{\Sigma}$. Note, that vi also tested $\bm{\Sigma}_{12}=0$ and $\bm{\Sigma}_{14}=0$ but since the estimated power curves are similar, we only present one of them here. The estimated power curve for $H$ can be found in Figure \ref{fig:PowerPlotNorm } and \ref{fig:powerplotT} for the Gaussian and t-distribution, respectively. We conclude, that we need a much higher number of levels in the multivariate t-distribution, as expected. This is a result of the heavier tail and the fact that the distribution of the test statistic is only asymptotic where the distribution is exact in the Gaussian case. 

\begin{figure}[htbp]
	\centering
\includegraphics[scale=0.8]{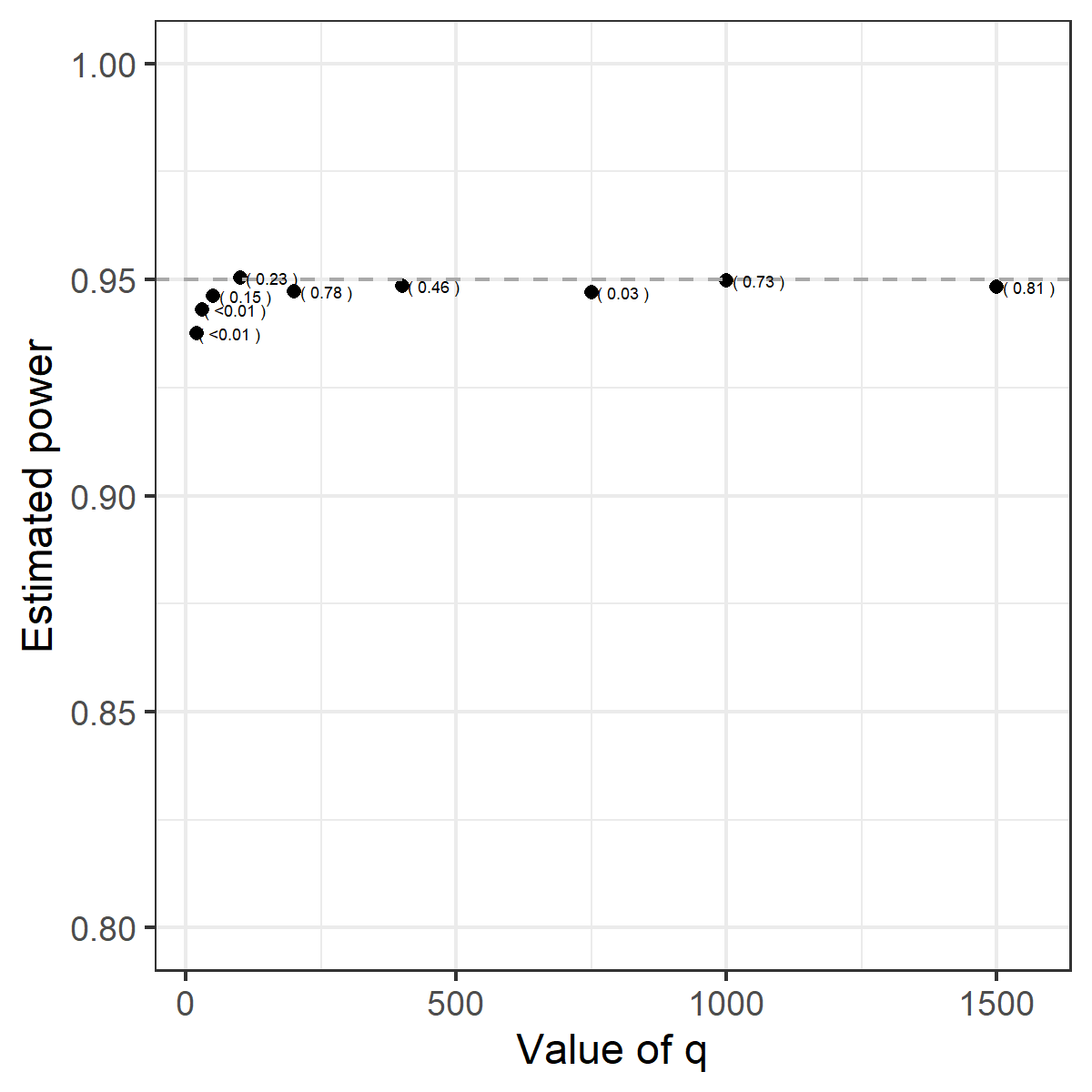}
	\caption{Estimated power as a function of $q$ based on $\bm{x}_1,\ldots,\bm{x}_{q}$. The values in parenthesis is p-values from a Kolmogorov Smirnoff test for a Uniform distribution. When the hypothesis is true and the distributional assumption is correct, the p-values should be uniformly distributed on the interval from zero to one. }
	\label{fig:PowerPlotNorm }
\end{figure}

\begin{figure}[htbp]
	\centering
	\includegraphics[scale=0.8]{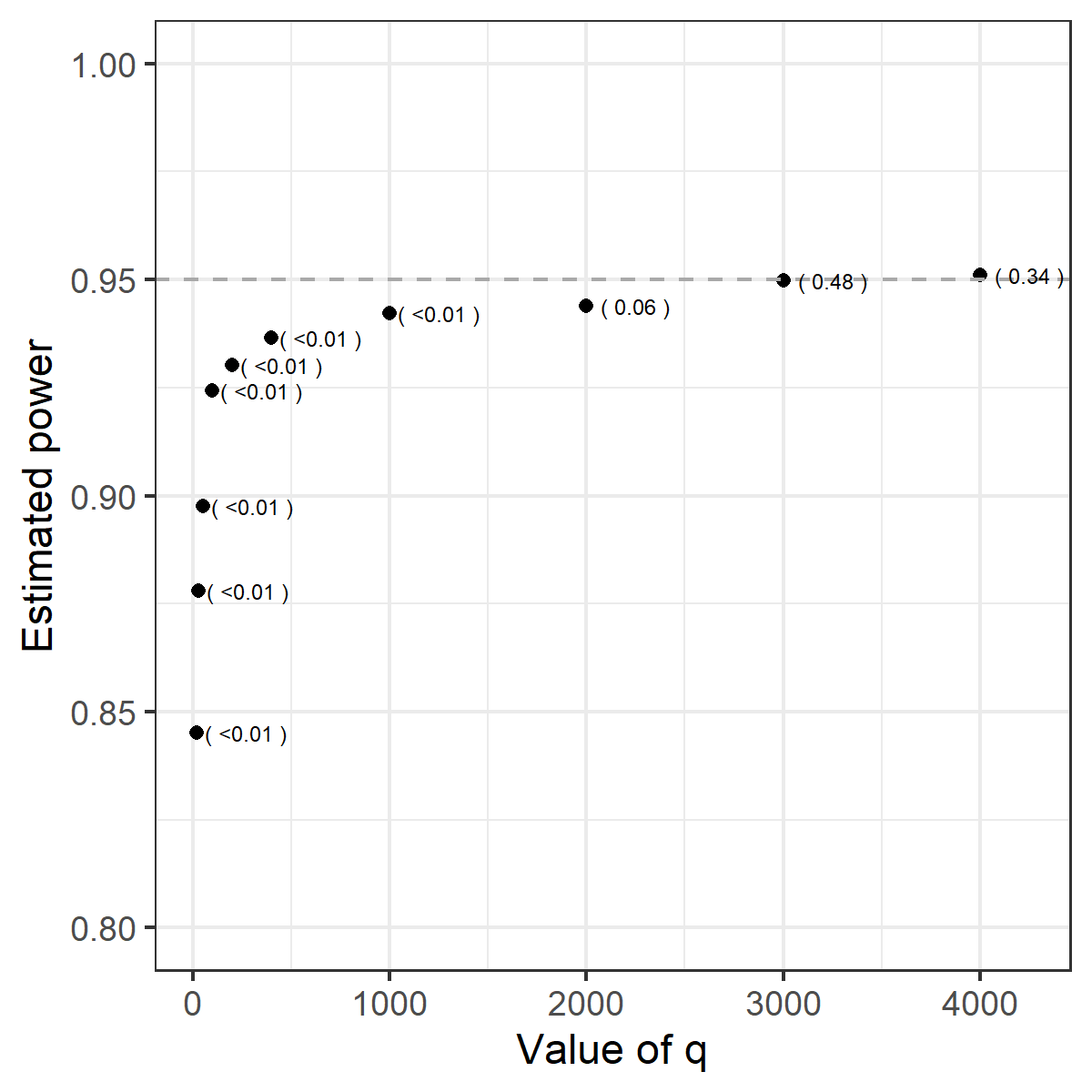}
\caption{Estimated power as a function of $q$ based on $\bm{y}_1,\ldots,\bm{y}_{q}$. The values in parenthesis is p-values from a Kolmogorov Smirnov test for a Uniform distribution. When the hypothesis is true and the distributional assumption is correct, the p-values should be uniformly distributed on the interval from zero to one.}
\label{fig:powerplotT}
\end{figure}

\FloatBarrier

\subsubsection{Graphical Representation of the Latent Covariance 
Structure in the Multivariate Model using a Statistical Test}

We can draw conclusions regarding the covariance structure of the random components by applying the tests described above to the 
estimated covariance matrix of the random components in the multivariate model described in Section~\ref{subsec:ModelDef}, and by
using the theory of graphical models as described in Section~\ref{subsec:GraphM}. 

Under the model formulated in Section~\ref{subsec:ModelDef}, we can estimate the covariance matrix of the random components consistently based on consistent predictions of the random components by applying 
Proposition~\ref{thm:convCovEst}. In this case, the distribution of the above-described tests will be asymptotically for the number of levels, $q$, and the number of observations, $n$, increasing.
In the case where we use the asymptotically approximately maximum likelihood estimator of the covariance matrix as described in  
Section~\ref{subsubsec:MLgaus} and \ref{subsubsec:MLellip}, the estimator is also consistent. Therefore, the tests still apply asymptotically. The same applies to another consistent estimate of the covariance matrix.

 We can examine the latent covariance structure in general by testing if the value 
 of each off-diagonal entry in the conditional covariance matrix is equal to zero. If 
 the p-value (possibly corrected for multiple 
 testing) is below a given significance level, we connect the corresponding nodes 
 by and edge. After constructing an undirected graph, 
 we can combine the undirected independence graph with the 
 responses as described in Section~\ref{subsec:GraphM}.
On the other hand, it might be of interest to test for a specific covariance 
structure of the latent variables. Here, the number of tests can be reduced using 
the structure of graphical models. It is possible to apply the test for independence 
between different groups of variables, without conditioning on a separating set, to 
test for independence between the isolated subgraphs in the graph (if any).  

\FloatBarrier
\subsection{Simulation Studies\label{subsec:SimStudy} }
	In this section, we perform a simulation study to examine the power of the two 
types of tests introduced in this paper under
multivariate generalised linear mixed models in the case of Gaussian and 
t-distributed random 
components. 
We simulate a two dimensional generalized linear mixed model with the 
conditional distributions being Gamma and Poisson, respectively.
We use a logarithm link function in both marginal models. Since our primary 
interest in this simulations study is the covariance structure of the random 
components, we will simulate a  model only including a constant in the fixed 
effects (the value of this constant was set to $0.6$).   The data was simulated 
with 
the length of the vectors of random components being 
$q=800$ (corresponding to $800$ experimental units or clusters), and with $40$ 
replicates for each unit giving $32\,000$ observations. 
Three models were simulated with different distributional assumptions for the 
random components (all having expectation zero), \ie a multivariate 
Gaussian, a multivariate  t-distribution with $11$ degrees of freedom, and a 
multivariate t-distribution with $7$ degrees of freedom. We estimated the 
power (probability of rejecting the null hypothesis)
of the tests for the different models on a grid of values for the 
off-diagonal entry
in the covariance matrix of the random components representing the same 
experimental unit. The covariance matrix is given by
\begin{align*}
	\bm{\Sigma} = \begin{Bmatrix}
		0.8166 & \sigma_{12}\\
		\sigma_{12} & 0.91302
	\end{Bmatrix},
\end{align*}
where $\sigma_{12}$ is varied on the grid $G=(0,0.02,0.04,0.1)$.

In each round of the simulation study, we test the hypothesis
\begin{align*}
	H_0: \sigma_{12}=0,
\end{align*}
and the resulting p-values are used to estimate the power for each point in $G$. 
Notice, that in the Gaussian case, $H_0$ implies independence, whereas, in the 
elliptical 
case it implies un-correlation.
We limit ourself to a two dimensional model partly because 
of the computational time and the preference of a high dimension of the vector 
of random components (as we saw in \ref{sec:simStudConvRat}, we need a high 
number of levels of the random components when 
these are assumed to be multivariate t-distributed), but also because it is 
difficult to control that a high dimensional covariance matrix stay positive 
definite when 
changing the off-diagonal values.

We would expect that the probability of rejecting the null hypothesis increases
when the corresponding entry in the covariance matrix is moved away 
from zero.  For each grid point in $G$, the model was simulated $500$ times 
and a
p-value for testing $H_0$ was calculated for each simulation. 
Thus, for each grid 
point, the probability of rejecting the null hypothesis
could be estimated based on the p-values.
Figure~\ref{fig:powerCurve} shows the estimated probabilities of rejecting the 
hypothesis (at a significance level of five percent) as a function of the 
off-diagonal value in 
the covariance matrix for each combination of model and test. 
From the figure, we conclude that when the random components are 
Gaussian distributed both tests reach the correct significance levels under the 
null hypothesis. However, the curve for the Gaussian test is steeper than the 
elliptical test in the part close to zero meaning that the test has a higher power 
to detect small deviations from the null hypothesis. In the case of t-distributed 
random components, the test based on normality rejects to often under the null 
hypothesis which lead to a power curve with a higher intersection with the 
y-axis. Moreover, we see that the shape of the curve differs from the power 
curve for the elliptical test. This result imply that it would be preferable to use 
the elliptical test in cases where the normality of the random components are 
uncertain.

We would expect, that the p-values follow a uniform distribution on zero to one 
under the hypothesis $H_{0}$.
In Figure~\ref{fig:qqPlotSimStud}, we present a Q-Q plot of the observed 
quantiles of  the calculated p-values versus the theoretical uniform quantiles 
based on $500$ simulations for each model and for each test. Recall, that we 
simulated three different models, where the random components followed 
either a multivariate 
Gaussian, a multivariate t-distribution with $11$ degrees of freedom or a 
multivariate t-distribution with $7$ degrees of freedom. For each model, we 
compared two 
different tests: a test based on normality and a test based on a general 
elliptically contoured distribution. The number added to each plot is the 
resulting p-values from a Kolmorogov-Smirnov test comparing the empirical 
distribution with the uniform distribution. As expected, the test based on an 
assumption of normality performs badly for the models where the random 
components follow a t-distribution.

\begin{figure}
	\centering
	\includegraphics[scale=0.6]{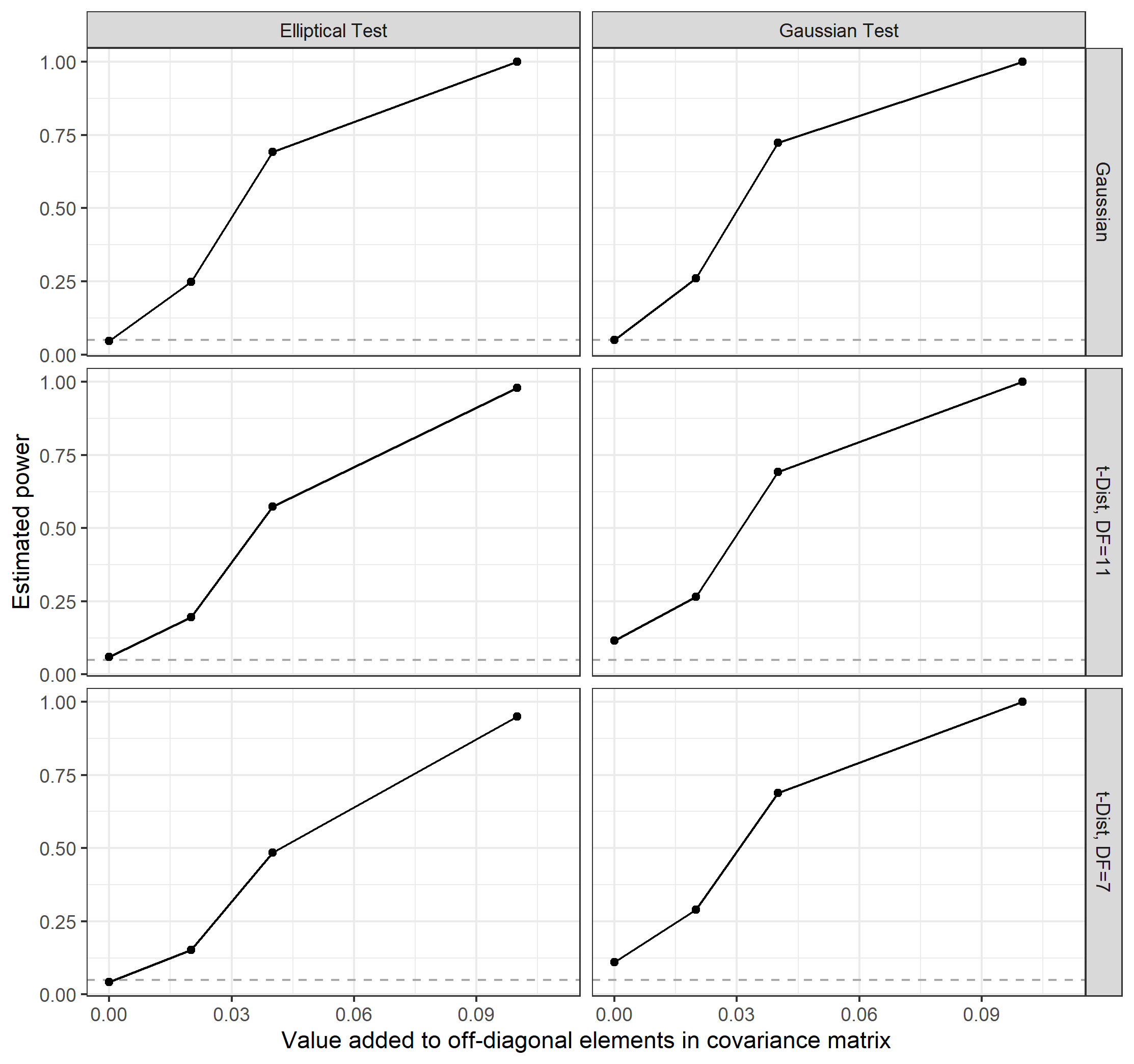}

	\caption{Power curves showing the estimated probability of rejecting the 
		hypotheses $H_{0}$ as a function of the off diagonal values in the 
		covariance 
		matrix for each test and distributional model for the random components. 
		\label{fig:powerCurve}}
\end{figure}

\begin{figure}
	\centering
	\includegraphics[scale=0.6]{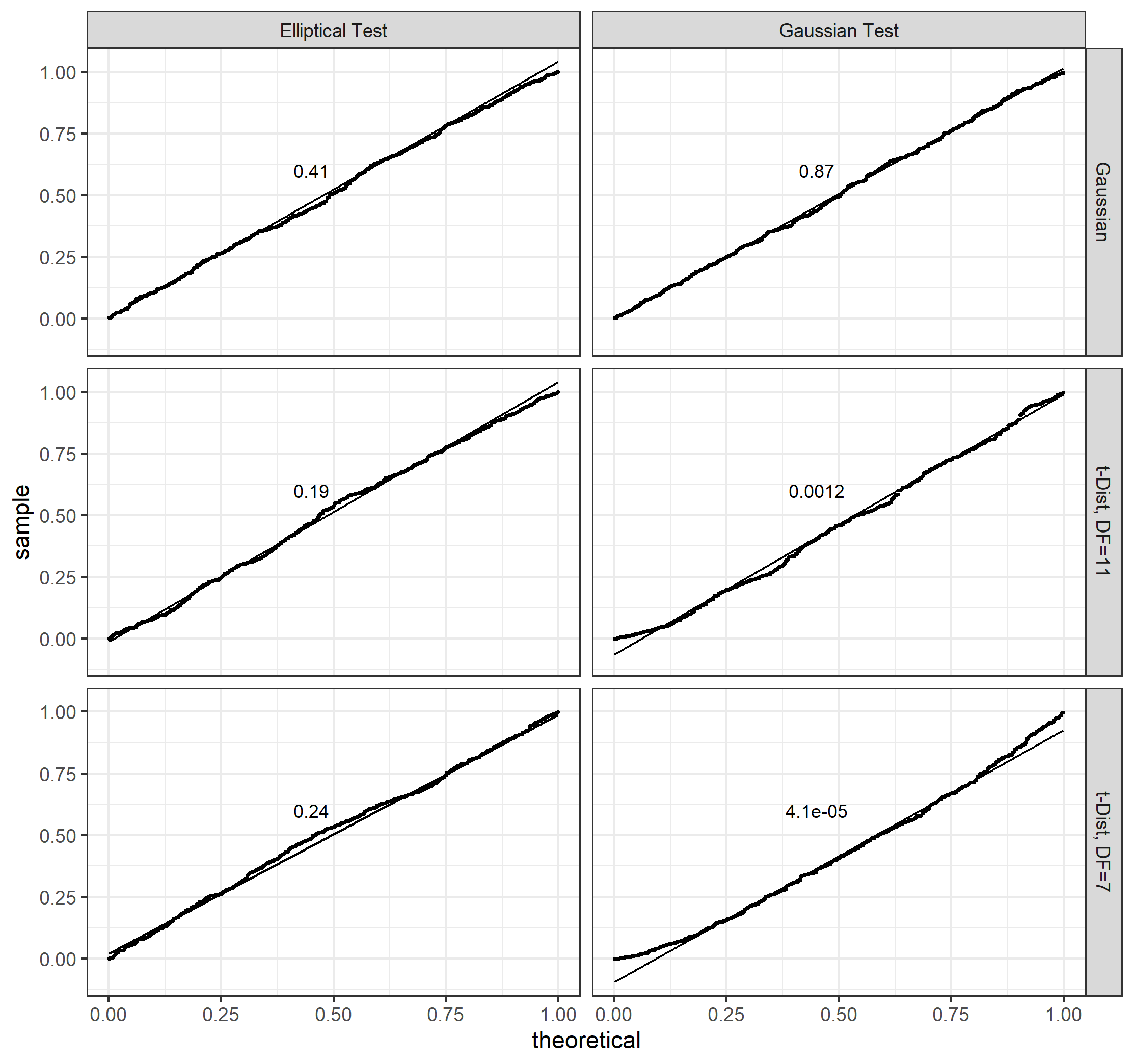}

	\caption{QQ-plots showing the quantiles of the p-values under the null 
		hypothesis against the theoretical quantiles of the uniform distribution for 
		each 
		test and distributional model for the random components. We 
		used the Kolmogorov-Smirnov test comparing the empirical distribution with 
		the uniform distribution. The resulting p-values are shown in the plots. 
		\label{fig:qqPlotSimStud} }
\end{figure}

\section{Discussion and Conclusion\label{sec:Con}}
The method for studying the dependence structure of multivariate responses described in this paper combines MGLMMs with the theory of graphical models and a variation of the tests for correlation and conditional correlation described in \cite{Anderson2003}. We constructed the MGLMMs used in this paper by joining marginal GLMMs that are based on weaker assumptions as compared to the literature (e.g., \citeauthor{Breslow1993}, \citeyear{Breslow1993}, \citeauthor{Mcculloch2001}, \citeyear{Mcculloch2001}, \citeauthor{McCulloch1997}, \citeyear{McCulloch1997}). Indeed, we do not assume the random components to be multivariate normally distributed. Moreover, we use dispersion models (which includes exponential dispersion models as a particular case) to define the conditional distributions of the responses given the random components. While \cite{Pelck2021A} developed techniques for estimating fixed effects and predicting random components of those MGLMMs, we concentrated here on the construction of methods for studying the correlation structure of multivariate responses.  The nature of the tests we used here forced us to restrict the distribution of the random components to be regular elliptically contoured distributions (including the multivariate normal distribution),  which is less general than the class of distributions of the random components used in \cite{Pelck2021A}. Still, the assumptions on the distribution of the random components used here are weak and yield a flexible class of MGLMMs. For example, we can use models with multivariate t-distributed random components, which have heavier tails than Gaussian random components. Our simulation studies demonstrate that MGLMMs with Gaussian random components might perform poorly in the presence of heavy tail distributions (\eg the significance level of the correlation tests becomes systematically wrong).

Multivariate generalised linear models (and MGLMMs) can be constructed by connecting several marginal generalised linear models (or GLMMs) using copulas. For instance \cite{Song2009} use Gaussian copulas for constructing multivariate dispersion models. While this approach might be fruitful in some contexts, it cannot be directly applied in the type of analysis we discuss in this paper because the distribution of the random components after applying the copula transformation are in general not elliptically contoured and therefore the tests we use here are applicable. 

Remarkably, the proposed test for elliptically contoured distributed random components does not depend on the choice of the elliptically contoured distribution used. Indeed, the test statistic of those tests depends only on the estimate of the covariance matrix. Therefore, we might view this test as a semiparametric test since the class of regular elliptically contoured distributions is not finite-dimensional. Naturally, the test based on the multivariate normal distribution is advantageous relative to the generic test based on elliptically contoured distributions when the random components are Gaussian distributed. We illustrate this claim in the simulation studies presented.

The inferential techniques described in this paper were applied in several fields recently. For instance, in \cite{pelck2020C} the method described above was used in a study of a system for monitoring the development of roots over time,  which involved binomial, and Poisson distributed
responses. Another example is presented in \cite{Pelck2021F} where our methods were applied to study the dependence structure of responses representing the development of a fungal disease and the concentration of volatile organic compounds. Those responses were modelled by \cite{Pelck2021F} using Gamma, binomial and compound Poisson families of distributions. Furthermore, in a third study, \cite{pelck2021D} used the methods studied here to discuss the covariance structure of the students' marks obtained in different admission exams at the University (Gaussian distributed) and the number of attempts required to pass the course of geometry (a Cox proportional model with discrete-time). Those examples illustrate the usefulness of the statistical tools studied in this paper. 

\section{Acknowledgement }
The authors were partially financed by the Applied Statistics Laboratory (aStatLab) 
at the Department of Mathematics, Aarhus University.

\bibliography{bibtexMaster}

\appendix

\section{Appendix}
\subsection{Estimation of Covariance Matrix
	\label{subsec:covEst}}
The methods that we will present in this paper rely on either an estimate of the 
covariance matrix proportional to the maximum likelihood estimate or a 
consistent estimate.
In this section, we discuss how such an estimate can be obtained based on 
consistent predictions of the random components. Such predictions can be 
obtained using the inference method described in \cite{Pelck2021A}.

A consistent estimate (for $n$ and $q$ increasing) of the covariance matrix can 
be found by calculating the sample covariance of the predicted values as we will 
see in Proposition~\ref{thm:convCovEst}.

In the case where we only have few cluster, \ie $q$ is small, we suggest a method 
to obtain an approximated maximum likelihood estimate of the covariance matrix. 
Here, we consider the 
general case were the random components follow an elliptically contoured 
distribution, and the special case where this distribution is assumed to be 
multivariate Gaussian separately.

\begin{prop} \label{thm:convCovEst}
	Consider the model described in Section~\ref{subsec:ModelDef}.
	For $j=1,\ldots,q$, let $\hat{\bm{b}}_j^{n}$ denote a $d$-dimensional vector 
	of predicted values of the random components corresponding to the $q\tth$ 
	cluster, $\bm{B_j}$, based on 
	at least $n=\min\{n_1,\ldots,n_q\}$ observations. Moreover, assume that
	\begin{align*}
		\hat{\bm{b}}_j^{n} \overset{P}{\longrightarrow} \bm{B }_j \quad \text{for 
		} 
		n \to \infty.
	\end{align*}
	Then,
	\begin{align*}
		\hat{\bm{\Sigma}}_q = \frac{1}{q-1}\sum_{j=1}^q 
		(\hat{\bm{b}}_j^{n_j}-\bar{\hat{\bm{b}}}_q)(\hat{\bm{b}}_j^{n_j}-\bar{\hat{\bm{b}}}_q)^T
		\, \overset{P}{\longrightarrow} \, \bm{\Sigma} \quad \text{for } q,n \to \infty,
	\end{align*}
	where $\bar{\hat{\bm{b}}}_q^n=\frac{1}{q}\sum_{j=1}^q 
	\hat{\bm{b}}_j^{n}$.
\end{prop}

\begin{proof}
	The proof follows from the fact that the predicted values of the random 
	components are consistent, the continuity of the sample covariance mapping 
	and 
	that the average converges to the expectation for $q$ increasing.
\end{proof}

We present below an approximation of the maximum likelihood function for 
estimating $\bm{\Sigma}$ based on the predicted values of the 
random components in the case of a multivariate Gaussian distribution, which can 
be used to estimate $\bm{\Sigma}$.

\subsubsection{Approximated maximum likelihood for estimating 
	\texorpdfstring{$\bm{\Sigma}$}{\Sigma} in the case of  Gaussian random 
	components\label{subsubsec:MLgaus}}

Consider the model described in Section~\ref{subsec:ModelDef}, where we 
assume that $\bm{B}_1,\ldots,\bm{B}_q$ are i.i.d Gaussian distributed with 
expectation zero and covariance matrix $\bm{\Sigma}$.
We let for $j=1,\ldots,q$, $\hat{\bm{b}}_j^{n}$ denote a $d$-dimensional 
vector 	of predicted values of the random components corresponding to the 
$j\tth$ cluster, $\bm{B_j}$, based on at least $n=\min\{n_1,\ldots,n_q\}$ 
observations. Moreover, we assume that the predicted values are conditional 
asymptotically Gaussian distributed (as in \cite{Pelck2021A}) for $n$ 
increasing given $\bm{B_j}=\bm{b}_j$ with conditional expectation 
$\bm{b}_j$ and covariance matrix $\bm{V}_{j}$.
Notice, that by the model assumptions, $\bm{V}_{j}$ is a diagonal matrix.

When $q$ is small, we can maximise the following with respect to 
$\bm{\Sigma}$ by inserting an estimate $\hat{\bm{V}}_{j}$ of $\bm{V}_{j}$:
\begin{align*}
	L(\bm{\Sigma}; \hat{\bm{b}}_1^{n},\ldots,\hat{\bm{b}}_1^{n}) &=
	\prod_{j=1}^q \int_{\R^d} \varphi(\bm{b}_j;\bm{\Sigma}) 
	h(\hat{\bm{b}}_j^{n};\bm{b}_j,\hat{\bm{V}}_{j}) d\bm{b}_j\\
	&=\prod_{j=1}^q \int_{\R^d} \vert 2\pi \bm{\Sigma}\vert^{-1/2} \exp 
	\left(-\tfrac{1}{2}\bm{b}_j^T\bm{\Sigma}^{-1}\bm{b}_j \right)\\
	&\quad \quad \quad		 \vert 2\pi \hat{\bm{V}}_{j}\vert^{-1/2}\exp \left 
	(-\tfrac{1}{2}(\bm{b}_j-\hat{\bm{b}}_j^{n})^T\hat{\bm{V}}_{j}^{-1}(\bm{b}_j-\hat{\bm{b}}_j^{n})\right)
	d\bm{b}_j\\
	&=\prod_{j=1}^q (2\pi)^{-d} \vert \bm{\Sigma} 
	\hat{\bm{V}}_{j}\vert^{-1/2} 
	\exp \left 
	(-\tfrac{1}{2}(\hat{\bm{b}}_j^{n})^T\bm{\Sigma}^{-1}\hat{\bm{b}}_j^{n}\right
	)\\
	&\quad \int_{\R^d} \exp \left 
	(\bm{b}_j^T\hat{\bm{V}}_{j}^{-1}\hat{\bm{b}}_j^{n}\right )
	\exp \left 
	(-\tfrac{1}{2}\bm{b}_j^T(\hat{\bm{V}}_{j}^{-1}+\bm{\Sigma}^{-1})\bm{b}_j
	\right 
	) d\bm{b}_j\\
	&=\prod_{j=1}^q	(2\pi)^{-d}  \vert \bm{\Sigma} 
	\hat{\bm{V}}_{j}\vert^{-1/2} 
	\exp \left 
	(-\tfrac{1}{2}(\hat{\bm{b}}_j^{n})^T\bm{\Sigma}^{-1}\hat{\bm{b}}_j^{n} 
	\right )\\
	&\quad \vert 2\pi 
	(\hat{\bm{V}}_{j}^{-1}+\bm{\Sigma}^{-1})^{-1}\vert^{1/2}
	\exp \left (\tfrac{1}{2} (\hat{\bm{V}}_{j}^{-1}\hat{\bm{b}}_j^{n})^T
	(\hat{\bm{V}}_{j}^{-1}+\bm{\Sigma}^{-1})^{-1}\hat{\bm{V}}_{j}^{-1}\hat{\bm{b}}_j^{n}
	\right)\\
	&=\prod_{j=1}^q	(2\pi)^{-d/2} 
	\vert \bm{\Sigma} 
	\hat{\bm{V}}_{j}\vert^{-1/2} 
	\vert (\hat{\bm{V}}_{j}^{-1}+\bm{\Sigma}^{-1})\vert^{-1/2}\\
	&\quad 
	\exp\left (-\tfrac{1}{2}(\hat{\bm{b}}_j^{n})^T
	[\bm{\Sigma}^{-1}-(\hat{\bm{V}}_{j}+\hat{\bm{V}}_{j}\bm{\Sigma}^{-1}\hat{\bm{V}}_{j})^{-1}]
	\hat{\bm{b}}_j^{n}\right ).
\end{align*}

\subsubsection*{Approximated maximum likelihood for estimating 
	\texorpdfstring{$\bm{\Sigma}$}{\Sigma} in the case of  general elliptical 
	contoured random components \label{subsubsec:MLellip}}

Consider the model described in Section~\ref{subsec:ModelDef}, where we 
assume that $\bm{B}_1,\ldots,\bm{B}_q$ are i.i.d elliptically contoured 
distributed with 
expectation zero and covariance matrix $\bm{\Sigma}$ for a given choice of 
the function $h$ in \eqref{eq:denEllip}.
We let for $j=1,\ldots,q$, $\hat{\bm{b}}_j^{n}$ denote a $d$-dimensional 
vector 	of predicted values of the random components corresponding to the 
$j\tth$ cluster $\bm{B_j}$, based on at least $n=\min\{n_1,\ldots,n_q\}$ 
observations. Moreover, we assume that the predicted values are conditional 
asymptotically Gaussian distributed (as in \cite{Pelck2021A}) for $n$ 
increasing given $\bm{B_j}=\bm{b}_j$ with conditional expectation 
$\bm{b}_j$ and covariance matrix $\bm{V}_{j}$.
Notice, that by the model assumptions, $\bm{V}_{j}$ is a diagonal matrix.

When $q$ is small, we can maximise the following with respect to 
$\bm{\Sigma}$ by inserting an estimate $\hat{\bm{V}}_{j}$ of $\bm{V}_{j}$ 
and using a Gaussian Hermite approximation of the integral:
\begin{align*}
	L(\bm{\Sigma}; \hat{\bm{b}}_1^{n},\ldots,\hat{\bm{b}}_1^{n}) &=
	\prod_{j=1}^q \int_{\R^d} \varphi(\bm{b}_j;\bm{\Sigma}) 
	h(\hat{\bm{b}}_j^{n};\bm{b}_j,\hat{\bm{V}}_{j}) d\bm{b}_j\\
	&= 
	\prod_{j=1}^q \int_{\R^d} \varphi(\bm{b}_j;\bm{\Sigma}) \vert 2\pi 
	\hat{\bm{V}}_{j} 
	\vert^{-1/2} \exp\big(-\tfrac{1}{2} (\bm{b}_j-\hat{\bm{b}}_j^{n})^T 
	\bm{V}_{j}^{-1}  (\bm{b}_j-\hat{\bm{b}}_j^{n})\big ) d\bm{b}_j\\
	&=
	\pi^{-d/2}  
	\prod_{j=1}^q \int_{\R^d} 
	\varphi(\sqrt{2}\hat{\bm{V}}_{j}^{1/2}\tilde{\bm{b}}_j+\hat{\bm{b}}_j^{n};\bm{\Sigma})
	\exp\big(-\tilde{\bm{b}}_{j}^T\tilde{\bm{b}}_{j} \big ) d\tilde{\bm{b}}_j\\
	&\approx
	\pi^{-d/2}  
	\prod_{j=1}^q \sum_{\bm{k}\in \mathcal{K}} 
	(w_{k_1} \ldots 
	w_{k_d})\varphi(\sqrt{2}\hat{\bm{V}}_{j}^{1/2}{\bm{x}}_{\bm{k}}+
	\hat{\bm{b}}_j^{n};\bm{\Sigma}), \\		
\end{align*}
with $\mathcal{K}=\{1,\ldots,l\}^d$, 
$\bm{w}_{\bm{k}}=(w_{k_1},\ldots,w_{k_d})$, 
$\bm{x}_{\bm{k}}=(x_{k_1},\ldots,x_{k_d})$, where $x_{k_j}$ denotes the 
$k_j\tth$ root of the Hermite polynomial with $l$ nodes and $w_{k_j}$ is 
the associated weight.

\subsection{Density of $V$ in Case of Gaussian Random Components 
\label{App:testVden}}
In this section we present a formula for the density of the distribution of $V$ 
defined in Equation~\eqref{eq:testV}. 
Recall, that $V$ is distributed according to
\begin{align*}
	V \sim \prod_{i=2}^{k-1}\prod_{j=1}^{d_i} Z_{ij},
\end{align*}
where the random variables $Z_{21},\ldots,Z_{(k-1)d_{k-1}}$ are independent 
and $Z_{ij}\sim \text{Beta}  \big (\tfrac{1}{2}[q-\bar{d}_i-j], \tfrac{1}{2} 
\bar{d}_{i} \big)$ with $\bar{d}_i=d_1+\ldots+d_{i-1}$ for $i=2,\ldots,(k-1)$ 
and $j=1,\ldots,d_i$.

Let $d=d_2+\ldots+d_{k-1}$. We adapt the notation in \cite{Tang1984} to 
obtain the density of $V$. Define $t_1=(2,1),  \ldots, t_{d_2}=(2,d_2), 
t_{d_2+1}=(3,1), \ldots, t_{d_2+d_3}=(3,d_3),\ldots, 
t_{d_2+\ldots+d_{k-2}+1}=(k-1,1),  \ldots,$\\  
$t_{d_2+\ldots+d_{k-1}}=(k-1,d_{k-1})$.
The density of $V$ can then be formulated as
\begin{align}
	f_{V}(v)=K_{d}v^{b(d)-1}(1-v)^{h(d)-1}\sum_{r=0}^{\infty} 
	\sigma_r^{(d)}(1-v)^r \quad \text{ for } 0<v<1,\label{eq:prodBetaDen}
\end{align}
where
\begin{align*}
	K_d&=\prod_{j=1}^d  \frac{\Gamma [c(j)]}{\Gamma [b(j)] },\\
	h(d)&=\sum_{j=1}^d[c(j)-b(j)]  \quad \text{for } 
\end{align*}
with $\Gamma(\cdot)$ denoting the Gamma function, $b(j)= 
\tfrac{1}{2}(q-\bar{d}_{t_j^{(1)}}-t_j^{(2)})$ and $c(j)=\tfrac{1}{2} 
\bar{d}_{t_j^{(1)}} + b(j)$ for $t_j=(t_j^{(1)},t_j^{(2)})$. The term 
$\sigma^{(d)}_r$ can be calculated by the recursive relation:
\begin{align*}
	\sigma_r^{(j)}=\frac{\Gamma[h(j-1)+r]}{\Gamma[h(j)+r]}\sum_{s=0}^r \Big 
	[ \frac{c(j)-b(j-1)}{s!} \sigma_{r-s}^{(j-1)} \Big ], \, r=0,1,2,\ldots, \,\, 
	j=2,3,\ldots,d, 
\end{align*}
with initial values $\sigma_0^{(1)}=(\Gamma[h(1)])^{-1}$ and 
$\sigma_r^{(1)}=0$ for $r=1,2,\ldots$.
Notice, that 
\begin{align*}
	\underset{j\to \infty}{\lim} \frac{\Gamma[h(j-1)+r]}{\Gamma[h(j)+r]} =0,
\end{align*}
such that the infinite sum in  \eqref{eq:prodBetaDen} can be truncated after some 
point. 

\end{document}